\documentclass{article}
\pdfoutput=1
\usepackage[ascii]{inputenc}
\usepackage[bookmarksnumbered,hypertexnames=false,colorlinks=true,linkcolor=blue,urlcolor=blue,citecolor=blue,anchorcolor=green,breaklinks=true,pdfusetitle]{hyperref}
\usepackage{amsmath,amssymb,amsthm}
\usepackage[margin=2cm]{geometry}
\usepackage[T1]{fontenc}
\usepackage[USenglish]{babel}
\usepackage[capitalize]{cleveref}

\newcommand{\IC}{\mathbb C}
\newcommand{\tensor}{{\textstyle\bigotimes}}
\newcommand{\la}{\lambda}
\newcommand{\aS}{\mathfrak S}
\newcommand{\VNP}{\textup{VNP}}
\newcommand{\VP}{\textup{VP}}
\DeclareMathOperator{\SL}{SL}
\DeclareMathOperator{\SU}{SU}
\DeclareMathOperator{\GL}{GL}
\DeclareMathOperator{\Sym}{Sym}
\DeclareMathOperator{\per}{per}
\DeclareMathOperator{\sgn}{sgn}
\DeclareMathOperator{\Pf}{Pf}

\swapnumbers
\numberwithin{equation}{section}
\newtheorem{proposition}[equation]{Proposition}

\newtheorem{theorem}[equation]{Theorem}
\newtheorem{conjecture}[equation]{Conjecture}
\theoremstyle{definition}
\newtheorem{definition}[equation]{Definition}

\begin{document}
\title{Hyperpfaffians and Geometric Complexity Theory}
\author{Christian Ikenmeyer\thanks{University of Liverpool, christian.ikenmeyer@liverpool.ac.uk, supported by DFG grant IK 116/2-1}\ \ and Michael Walter\thanks{University of Amsterdam, m.walter@uva.nl, supported by NWO Veni grant 680-47-459}}
\date{December 19, 2019}
\maketitle

\begin{abstract}
The hyperpfaffian polynomial was introduced by Barvinok in 1995 as a natural generalization of the well-known Pfaffian polynomial to higher order tensors.
We prove that the hyperpfaffian is the unique smallest degree SL-invariant on the space of higher order tensors.
We then study the hyperpfaffian's computational complexity and prove that it is VNP-complete.
This disproves a conjecture of Mulmuley in geometric complexity theory about the computational complexity of invariant rings.
\end{abstract}

\section{Algebraic Complexity Theory: VP versus VNP}
The main model of computation that we are concerned with in this paper is the model of arithmetic circuits (see~\cite{BCS:97} or~\cite{Sap:17} for an introduction).
Throughout the paper we take the complex numbers~$\IC$ as our ground field.
An \emph{arithmetic circuit}~$C$ is a directed acyclic graph whose vertices of indegree 0 are labeled with variable names or constants from the ground field; all other vertices are labeled with either~``$+$'' or~``$\times$'', and there is the additional restriction that there is exactly one vertex of outdegree 0.
An arithmetic circuit~$C$ computes a polynomial at each vertex, by induction over the circuit structure.
We say that~$C$ \emph{computes} the polynomial that is computed at its outdegree 0 vertex.
The \emph{size} of an arithmetic circuit is its number of vertices.
For a multivariate polynomial~$f$ let~$L(f)$ denote the minimum size required to compute~$f$ with an arithmetic circuit.
We call~$L(f)$ the \emph{circuit complexity} of~$f$.

A sequence of natural numbers~$(c_i)_i$ is called \emph{polynomially bounded} if there exists a univariate polynomial~$p$ such that~$c_i\leq p(i)$ for all~$i$.
If, for a sequence of multivariate polynomials~$f_i$, the sequence of degrees and the sequence of numbers of variables of the $f_i$ is polynomially bounded, then we call $f=(f_i)_i$ a \emph{p-family}.
For example, the \emph{permanent polynomial family}
\[
\per_n := \sum_{\pi \in \aS_n} \prod_{i=1}^n x_{i,\pi(i)}
\]
is a p-family, where $\aS_n$ denotes the symmetric group on $n$ letters and the $n^2$ variables are doubly indexed.
The complexity class VP consists of all p-families $f=(f_i)_i$ with polynomially bounded circuit complexity $(L(f_i))_i$.

A \emph{projection} of a polynomial is defined as its evaluation at a point parametrized by affine linear polynomials.
For example, the polynomial $xz+xy+x+y+1$ is a projection of $\per_2$, because
\[
\per_2\begin{pmatrix}
       x & y+1\\
       x+1 & z\\
      \end{pmatrix} = xz+(x+1)(y+1) = xz+xy+x+y+1.
\]
Valiant proved that every polynomial is the projection of $\per_n$ for $n$ large enough~\cite{Val:79}.
The smallest $n$ required for this is called the \emph{permanental complexity} of a polynomial~$f$.
The class VNP is the class of p-families whose permanental complexity is polynomially bounded.

A p-family $f = (f_i)_i$ is called VNP-\emph{complete} if $f\in\VNP$ and there exists a polynomially bounded sequence~$c$ such that $\per_i$ is a projection of $f_{c_i}$.
If for a VNP-complete polynomial $f$ we have that $f \in \VP$, then $\VP=\VNP$.
Valiant famously conjectured that $\VP\neq\VNP$, which is the flagship conjecture of algebraic complexity theory and an algebraic analog of the famous $\text{P}\neq\text{NP}$ conjecture.
Geometric complexity theory is an approach towards proving that~$\VP\neq\VNP$.

\section{Invariant Rings and Symmetric Tensors}\label{sec:invariants}
Let $\SL_n$ denote the special linear group, consisting of complex $n\times n$-matrices with unit determinant.
It acts canonically on $\IC^n$ by matrix-vector multiplication.
This action extends to any $m$-th tensor power $\tensor^m \IC^n$ by
\[
g(v_1 \otimes \cdots \otimes v_m) := (g v_1) \otimes \cdots \otimes (g v_m)
\]
and linear continuation.
We will always use the standard inner product on $\tensor^m \IC^n$, which is invariant under the action of the subgroup~$\SU_n$ of~$\SL_n$.

Let~$W$ be an arbitrary finite-dimensional $\SL_n$-representation (such as $W = \tensor^m\IC^n$).
Then $\SL_n$ also acts on the vector space of homogeneous degree-$d$ polynomial functions~$\IC[W]_d$ on $W$, via the canonical pullback
\[
(g f)(w) := f(g^t w),
\]
as well as on the ring of all polynomial functions~$\IC[W] = \bigoplus_{d=0}^\infty \IC[W]_d$.
A function $f \in \IC[W]$ is called \emph{invariant} if $\forall g \in \SL_n$ we have $gf=f$.
The \emph{invariant ring} $\IC[W]^{\SL_n}$ of~$W$ is defined as the ring of all invariants in $\IC[W]$.
Hilbert's finiteness theorem implies that $\IC[W]^{\SL_n}$ is finitely generated.

It is convenient to identify polynomial functions with symmetric tensors.
Note that $\SL_n$ acts canonically on any $d$-th tensor power $\tensor^d W$ of~$W$.
This action restricts to the $d$-th symmetric tensor power~$\Sym^d W$, i.e., the $\aS_d$-invariant subspace of $\tensor^d W$.
For any~$t \in \Sym^d W$, we can define a homogeneous degree-$d$ polynomial~$f\in\IC[W]_d$ by~$f(w) := \langle t, w^{\otimes d}\rangle$.
Here we use the inner product on $\Sym^d W$ induced by an $\SU_n$-invariant inner product on~$W$.
Then, $f$ is invariant if and only if the symmetric tensor~$t$ is invariant, i.e., if $\forall g \in \SL_n$ we have~$gt=t$.
We will tacitly go back and forth between symmetric tensors in $\Sym^d \tensor^m \IC^n$ and homogeneous polynomials in~$\IC[\tensor^m \IC^n]_d$.

\section{Efficient Generators for Invariant Rings}
In~\cite{Mul:17}, Mulmuley proposes the study of the computational complexity of invariant rings.

\begin{definition}
Let $W$ be an $\SL_n$-representation of dimension $N$ and let $C=C(x_1,\ldots,x_N,y_1,\ldots,y_M)$ be an arithmetic circuit.
We say that $C$ \emph{succinctly encodes the generators of the invariant ring of $W$} if the set of polynomials $\{C(x_1,\ldots,x_N,\alpha_1,\ldots,\alpha_M) \mid \forall i: \alpha_i \in \IC\}$ are $\SL_n$-invariants that generate $\IC[W]^{\SL_n}$ as a ring.
\end{definition}

\noindent
Mulmuley states the following conjecture (see~\cite[Conj.~5.3]{Mul:17}, and also the rephrasing~\cite[Conj.~1.4]{GMOW:19}).

\begin{conjecture}\label{conj:mulmuley}
Let $W=(W_1,W_2,\ldots)$ be a sequence of representations, with $W_i$ an $\SL_{n_i}$-representation of dimension~$N_i$.
Then there exists a polynomial $p$ and a sequence $C_i$ of arithmetic circuits of size $\leq p(n_i N_i)$ such that $C_i$ succinctly encodes the generators of the invariant ring of $W_i$.
\end{conjecture}

\noindent
The following recent works make partial progress on \cref{conj:mulmuley} under the assumption that $\VP\neq\VNP$:

\begin{itemize}
\item In~\cite{BIJL:18}, $\SL_n$ was replaced with a symmetric group, and the conjecture was disproved.
 The conjecture is not explicitly mentioned there though.
\item In~\cite{GMOW:19}, $\SL_n$ was replaced by a product of special linear groups, and the conjecture was disproved.
\end{itemize}
The original conjecture for $\SL_n$ remained open and Oliveira stated its resolution as an open problem at the Reunion Workshop for the Lower Bounds in Complexity Theory program at the Simons Institute for the Theory of Computing, Berkeley, CA, on Dec.~11\textsuperscript{th}, 2019.

Here we disprove the original conjecture under the assumption that $\VP\neq\VNP$.
We use the same proof technique as in~\cite{BIJL:18} and~\cite{GMOW:19}:
Namely, it suffices to construct a sequence of representations such that each has a unique invariant of lowest degree and show that this family of invariants is $\VNP$-complete.
In the next section we achieve this for the $\SL_n$-representations $\tensor^{4\ell}\IC^n$ for any fixed $\ell$, which finishes the proof.

\section{Hyperpfaffians}
For even $n$, the \emph{Pfaffian} is the unique (up to scale) homogeneous $\SL_n$-invariant of degree $n/2$ on~$\IC^n\otimes\IC^n$.
There are no $\SL_n$-invariants in lower degrees.
If we identify $\IC^n \otimes \IC^n$ with the space of $n\times n$ matrices then the Pfaffian is invariant under the action of $\SL_n$ given by $g \cdot A := gAg^t$.
The defining property of the Pfaffian generalizes to tensors of even order as follows (the classical Pfaffian is the special case of $k=1$):

\begin{proposition}\label{pro:nolowerdegreeinv}
For any $k$ and any $n$ divisible by $2k$, there is a unique (up to scale) homogeneous $\SL_n$-invariant polynomial~$\Pf_{k,n}$ of degree~$\frac{n}{2k}$ on $\tensor^{2k} \IC^n$.
$\Pf_{k,n}$ identifies with the symmetric tensor $e_1 \wedge \cdots \wedge e_n \in \Sym^d \tensor^{2k}\IC^n$.
There are no non-constant $\SL_n$-invariants in lower degrees.
\end{proposition}

\noindent
Before proving \cref{pro:nolowerdegreeinv} we recall some representation theory.
A \emph{partition} $\la$ is a nonincreasing sequence of natural numbers with finite support.
We write $\la \vdash_n m$ to say that $|\la| := \sum_i \la_i = m$ and $\la_{n+1}=0$.
If $\la_{n+1}=0$, then we say that $\la$ is an $n$-partition.
The irreducible polynomial $\GL_n$-representations are indexed by $n$-partitions.
For a partition $\la$ let $\{\la\}$ denote the irreducible $\GL_n$-representation corresponding to~$\la$.
The representation $\{\la\}$ is trivial as an $\SL_n$-representation iff $\la_1=\ldots=\la_n$; note that this implies that $n \mid m$.
The irreducible representations of $\aS_m$ are indexed by partitions $\la$ with $|\la|=m$.
Let $[\la]$ denote the irreducible $\aS_m$-representation corresponding to~$\la$.
Schur-Weyl duality states that
\[ 
\tensor^m \IC^n = \bigoplus_{\la \vdash_n m} \{\la\}\otimes[\la].
\] 
Thus $\tensor^m \IC^n$ contains $\SL_n$-invariant vectors if and only if $n \mid m$.
For $m=n$, there is a unique (up to scale) $\SL_n$-invariant vector, since~$[1^n]$ is the one-dimensional sign representation of~$\aS_n$.
This vector is given by~$e_1 \wedge e_2 \wedge \cdots \wedge e_n$, where $a \wedge b := \frac 1 2 (a \otimes b - b \otimes a)$, and higher order wedge products are defined analogously.

\begin{proof}[Proof of \cref{pro:nolowerdegreeinv}]
It suffices to show that $\Sym^d \tensor^{2k} \IC^n$ contains no $\SL_n$-invariant vector if $0<d<\frac{n}{2k}$ and that it contains a unique such vector if $d=\frac{n}{2k}$.
Note that $\Sym^d \tensor^{2k} \IC^n$ is a subspace of $\tensor^d \tensor^{2k} \IC^n \simeq \tensor^{2kd} \IC^n$.
Thus the first claim holds since $\tensor^m \IC^n$ contains $\SL_n$-invariant vectors only if $n \mid m$, but $0<2kd<n$ if $0<d<\frac{n}{2k}$.
For~$d=\frac{n}{2k}$, $\tensor^d \tensor^{2k} \IC^n \simeq \tensor^n \IC^n$ contains the unique $\SL_n$-invariant vector $v = (e_1 \wedge \cdots \wedge e_{2k}) \wedge \cdots \wedge (e_{2k(d-1)} \wedge \cdots \wedge e_{2kd})$.
It remains to show that $v$ is symmetric, i.e., an element of~$\Sym^d \tensor^{2k} \IC^n$.
This holds since each of the~$d$ blocks has even size~$2k$ and the wedge product is graded-commutative.
This proves the second claim.
\end{proof}

\noindent
The polynomial~$\Pf_{k,n}$ was introduced in~\cite[Def.~3.4]{Bar:95} in its monomial presentation, where it is called the \emph{hyperpfaffian}.
Note that, for fixed $k$, $\Pf_k := (\Pf_{k,2k},\Pf_{k,4k},\dots)$ is a p-family, since~$\Pf_{k,n}$ has degree~$\frac{n}{2k}$ and $n^{2k}$ variables.
The monomial presentation in~\cite{Bar:95} immediately yields that $\Pf_k \in \VNP$.
\begin{theorem}
For even $k$, $\Pf_k$ is \VNP-complete.
\end{theorem}
\begin{proof}
For any $d$ and $n = 2kd$, we present a projection of $\Pf_{k,n}$ to the $d\times d$ permanent.
The same projection yields the determinant if $k$ is odd, which explains why the proof does not work for the classical Pfaffian ($k=1$).
The case~$k=2$ is enough to disprove Mulmuley's conjecture.

According to \cref{pro:nolowerdegreeinv}, the Pfaffian $\Pf_{k,n}$ identifies with the symmetric tensor
\begin{align*}
  v := e_1 \wedge \cdots \wedge e_n \in \Sym^d \tensor^{2k}\IC^n.
\end{align*}
Thus, the evaluation~$\Pf_{k,n}(p)$ at a tensor $p \in \tensor^{2k} \IC^n$ is given by $\langle v, p^{\otimes d}\rangle$ (cf.~\cite[Sec.~4.2(A)]{Ike:12}).
We choose
\[
p = \sum_{i,j=0}^{d-1} x_{i+1,j+1} (e_{1+2ki} \otimes e_{2+2ki} \otimes \cdots \otimes e_{k+2ki} \otimes e_{k+1+2kj} \otimes e_{k+2+2kj} \otimes \cdots \otimes e_{2k+2kj}),
\]
where the $x_{i,j}$ ($1 \leq i,j \leq d$) are formal variables.

The point $p$ is parametrized linearly by the $x_{i,j}$, so the evaluation of $\Pf_{k,n}$ at $p$ is a projection of $\Pf_{k,n}$.
We verify that the evaluation of $\Pf_{k,n}$ at $p$ gives the $d \times d$ permanent (up to a constant nonzero scalar) as follows.
\begin{align*}
\quad p^{\otimes d} &=\hspace{-.5cm}\sum_{i_1,j_1,\ldots,i_d,j_d=0}^{d-1}\hspace{-.5cm} x_{i_1+1,j_1+1} \cdots x_{i_d+1,j_d+1}
(e_{1+2ki_1} \otimes e_{2+2ki_1} \otimes \cdots \otimes e_{k+2ki_1} \otimes e_{k+1+2kj_1} \otimes e_{k+2+2kj_1} \otimes \cdots \otimes e_{2k+2kj_1})\\ &\quad\quad\quad\quad\quad\quad\quad\quad\otimes \cdots \otimes
(e_{1+2ki_d} \otimes e_{2+2ki_d} \otimes \cdots \otimes e_{k+2ki_d} \otimes e_{k+1+2kj_d} \otimes e_{k+2+2kj_d} \otimes \cdots \otimes e_{2k+2kj_d})
\end{align*}
and by linearity
\begin{align*}
\quad \langle v,p^{\otimes d}\rangle &= \hspace{-.5cm}\sum_{i_1,j_1,\ldots,i_d,j_d=0}^{d-1}\hspace{-.5cm} x_{i_1+1,j_1+1} \cdots x_{i_d+1,j_d+1}
\langle v,(e_{1+2ki_1} \otimes e_{2+2ki_1} \otimes \cdots \otimes e_{k+2ki_1} \otimes e_{k+1+2kj_1} \otimes e_{k+2+2kj_1} \otimes \cdots \otimes e_{2k+2kj_1})\\ &\quad\quad\quad\quad\quad\quad\quad\quad\otimes \cdots \otimes
(e_{1+2ki_d} \otimes e_{2+2ki_d} \otimes \cdots \otimes e_{k+2ki_d} \otimes e_{k+1+2kj_d} \otimes e_{k+2+2kj_d} \otimes \cdots \otimes e_{2k+2kj_d})\rangle
\end{align*}
A crucial property of $v$ is that $\langle v,e_{\pi(1)} \otimes e_{\pi(2)} \otimes \cdots \otimes e_{\pi(n)}\rangle \neq 0$ iff $\pi$ is a permutation, in which case it is equal to the sign of the permutation.
It follows that the nonzero summands in $\langle v,p^{\otimes d}\rangle$ are precisely those for which~$i=(i_1,\ldots,i_d)$ and $j=(j_1,\ldots,j_d)$ are permutations of $\{0,\ldots,d-1\}$.
For a single summand with $i$ and $j$ permutations we see:
\begin{align*}
&\quad x_{i_1+1,j_1+1} \cdots x_{i_d+1,j_d+1}
\langle v,(e_{1+2ki_1} \otimes e_{2+2ki_1} \otimes \cdots \otimes e_{k+2ki_1} \otimes e_{k+1+2kj_1} \otimes e_{k+2+2kj_1} \otimes \cdots \otimes e_{2k+2kj_1})\\ &\quad\quad\quad\quad\quad\quad\quad\quad\otimes \cdots \otimes
(e_{1+2ki_d} \otimes e_{2+2ki_d} \otimes \cdots \otimes e_{k+2ki_d} \otimes e_{k+1+2kj_d} \otimes e_{k+2+2kj_d} \otimes \cdots \otimes e_{2k+2kj_d})\rangle
\\&= \sgn(i)^k \sgn(j)^k x_{i_1+1,j_1+1} \cdots x_{i_d+1,j_d+1}.
\end{align*}
Hence, for even $k$ we obtain $\langle v,p^{\otimes d}\rangle = d! \per_d$.
\end{proof}

\bibliographystyle{alpha}
\bibliography{lit}

\end{document}